\documentclass[11pt]{article}
\usepackage[left=1in,top=1in,right=1in,bottom=1in,head=.1in,nofoot]{geometry}

\usepackage{setspace,url,bm,amsmath} % For double-spacing, URL font, math symbols

\usepackage{titlesec} % Section header formatting
\titlelabel{\thetitle.\quad} % Section header formatting
\titleformat*{\section}{\bf\large\center\uppercase} % Section header formatting

\usepackage{graphicx} % Graphics scaling
\usepackage{bbm}
\usepackage{latexsym}
\usepackage{caption}
\usepackage[margin=20pt]{subcaption}
\usepackage{hyperref}

\usepackage[table]{xcolor}

\newcommand{\GG}[1]{}

\usepackage{amsthm}
\usepackage{comment}
\theoremstyle{definition}
\newtheorem{assumption}{Assumption}
\newtheorem{theorem}{Theorem}

\newtheorem{corollary}{Corollary}

\usepackage{natbib} % ASA citation style
\bibpunct{(}{)}{;}{a}{}{,} % ASA citation style

\usepackage{etoolbox} % Bibliography underfull/overfull box fix
\apptocmd{\sloppy}{\hbadness 10000\relax}{}{} % Bibliography underfull/overfull box fix

\usepackage{color}
\usepackage{listings}

\begin{document}
\doublespacing
\title{\bf Construction of alternative hypotheses for randomization tests with ordinal outcomes}
\author{Jiannan Lu, Peng Ding and Tirthankar Dasgupta\footnote{Address for correspondence: Tirthankar Dasgupta, Department of Statistics, Harvard University, 1 Oxford Street, Cambridge, 02138 Massachusetts, U.S.A.
Email: \texttt{dasgupta@stat.harvard.edu}}\\Harvard University}
\date{}
\maketitle

\begin{abstract}
For ordinal outcomes, we construct sequences of alternative hypotheses in increasing departures from the sharp null hypothesis of zero treatment effect on each experimental unit, to help assess the powers of randomization tests in randomized treatment-control experiments.
\end{abstract}


\section{Introduction}

Introduced by \cite{Fisher:1935}, randomization tests are useful tools for causal inference, because they assess the statistical significance of treatment effects without making any assumptions about the underlying distribution of the outcome. Early theories on randomization tests were developed by \citet{Pitman:1938} and \citet{Kempthorne:1952}, which showed that many statistical procedures can be viewed as approximations of randomization tests. To quote \citet{Bradley:1968}, ``[a] corresponding parametric test is valid only to the extent that it results in the same statistical decision [as the randomization test].'' A crucial advantage of randomization tests is their abilities to handle non-standard (e.g., ordinal) outcomes. However, there appears to be limited research on how to assess the powers of randomization tests for ordinal outcomes.

Several researchers like \citet{Miller:2006}, have pointed out that in many randomized experiments, experimental units cannot be viewed as a random sample drawn from a hypothetical population. Therefore it is important to restrict the scope of inference to the finite population of experimental units. The potential outcomes framework \citep{Neyman:1923, Rubin:1974} makes randomization tests easy to interpret, and more importantly, helps us recognize its role in making finite-population inference. However, it does not naturally permit the assessment of the powers of randomization tests, which requires constructing alternatives to the sharp null hypothesis of zero treatment effect on each experimental unit. The existing literature \citep[e.g.,][]{Lehmann:1975, Rosenbaum:2010} assesses the powers of randomization tests by invoking infinite population models, primarily to circumvent the difficulty associated with construction of finite population alternatives. Such a construction requires specifying the potential outcomes for all experimental units, and is thus considered ``a thankless task'' by experts \citep{Rosenbaum:2010}. Such a ``thankless task'' can be made easy by invoking the assumption of independence of potential outcomes, as in some existing literature \citep[e.g.,][]{Cheng:2009, Agresti:2010}. However, the impact of association between potential outcomes on the powers of randomization tests in a finite population setting has never been investigated before.

In this paper, we demonstrate that it is indeed possible to construct alternative populations of ordinal potential outcomes without invoking the independence assumption. We propose a procedure to construct alternative hypotheses for ordinal outcomes, which is particularly useful in the finite population setting, but is also applicable to infinite populations. Moreover, unlike the existing literature \citep[e.g.,][]{Cheng:2009, Agresti:2010} which assume independent potential outcomes, our construction procedure takes into account the dependence structure of the potential outcomes and demonstrates (through simulation studies) that the association indeed does affect the powers of randomization tests in a finite population setting.

The paper proceeds as follows. Section \ref{sec:rev} reviews randomization tests of the sharp null hypothesis for ordinal outcomes. Section \ref{sec:mj} introduces two measures quantifying departures from the sharp null hypothesis, discusses their relationships to the powers of randomization tests, and proposes a procedure to construct alternative hypotheses in closed forms. Section \ref{sec:simu} reports the results of a simulation study that demonstrates how to use the proposed construction procedure to assess the powers of randomization tests. Section \ref{sec:dis} presents some concluding remarks.

\section{Randomization Tests for Ordinal Outcomes}\label{sec:rev}

\subsection{Potential Outcomes, Sharp Null Hypothesis and Randomization Test} \label{ss:pot}
We consider a completely randomized experiment with $N$ units, a binary treatment and an ordinal outcome with $J$ categories labeled as $0,\ldots,J-1$, where 0 and $J-1$ are the ``worst'' and ``best'' categories, respectively. Under the Stable Unit Treatment Value Assumption \citep{Rubin:1980} that there is only one version of the treatment and no interference among units, we define the pair $\left\{Y_i(1), Y_i(0)\right\}$ as the potential outcomes of the $i$th unit under treatment and control. Let
$$
p_{kl} = \mathrm{pr} \left\{ Y_i(1) = k, Y_i(0) = l \right\}
\quad
(k, l = 0, 1, \ldots, J-1)
$$
be the probability of potential outcomes $k$ and $l,$ under treatment and control. Here, the proportion or probability notation pr($\cdot$) can be defined for a finite population with $N$ units, or for an infinite population. The $J \times J$ probability matrix $\bm P = \left( p_{kl} \right)_{0 \le k,l \le J-1}$ summarizes the joint distribution of the potential outcomes, and plays a crucial role in our later discussion. Let
\begin{equation*}
p_{k+} = \sum_{l^\prime=0}^{J-1}p_{kl^\prime},
\quad
p_{+l} = \sum_{k^\prime =0}^{J-1}p_{k^\prime l}
\quad
(k, l = 0, 1, \ldots, J-1).
\end{equation*}
The vectors $\bm{p}_1 = \left( p_{0+}, \ldots, p_{J-1, +} \right)^\mathrm{T}$ and $\bm{p}_0 = \left( p_{+0}, \ldots, p_{+, J-1} \right)^\mathrm{T}$ characterize the marginal distributions of the potential outcomes under treatment and control.

Using the potential outcomes, we express the sharp null hypothesis as
$
Y_i(0)=Y_i(1)
$
for all $i.$ Under the sharp null hypothesis, the probability matrix ${\bm P}$ is diagonal with $p_{j+} = p_{jj} = p_{+j},$ for all $j=0, 1, \ldots, J-1.$ To test the sharp null hypothesis, we use data from completely randomized experiments with $N_1$ units assigned to treatment. For the $i$th unit, we denote its treatment indicator as $W_i,$ and its observed outcome is consequently
$
Y_i^{\textrm{obs}}=W_iY_i(1) + (1-W_i)Y_i(0).
$
For each $j,$ let $n_{0j}$ and $n_{1j}$ respectively represent the numbers of units exposed to control and treatment with observed outcome $j.$ Given the observed data, we first choose a suitable test statistic, typically a ``measure of extremeness" (\citealp{Brillinger:1978}), and then obtain a $p$-value by comparing the observed value of the test statistic to its randomization distribution.

\section{Construction of Alternative Hypotheses}\label{sec:mj}

To evaluate the powers of randomization tests, we need to construct alternatives to the sharp null hypothesis. We can violate the sharp null hypothesis in two distinct ways:
\begin{enumerate}
\item different marginal probabilities, i.e., $\bm p_1 \ne \bm p_0;$
\item identical marginal probabilities, but nonzero off diagonal elements in $\bm P$.
\end{enumerate}
For example, consider the following probability matrices
\begin{equation}\label{eq:examples}
\bm{P}_1  =
\left( \begin{array}{ccc}
\frac{1}{6} & \frac{1}{6} & \frac{1}{6} \\
0 & \frac{1}{6} & \frac{1}{6} \\
\frac{1}{6} & 0 & 0 \end{array}
\right),
\quad
\bm{P}_2  =
\left( \begin{array}{ccc}
0 & \frac{1}{6} & \frac{1}{6} \\
\frac{1}{6} & 0 & \frac{1}{6} \\
\frac{1}{6} & \frac{1}{6} & 0 \end{array}
\right),
\quad
\bm{P}_3  =
\left( \begin{array}{ccc}
\frac{1}{9} & \frac{1}{9} & \frac{1}{9} \\
\frac{1}{9} & \frac{1}{9} & \frac{1}{9} \\
\frac{1}{9} & \frac{1}{9} & \frac{1}{9} \end{array}
\right),
\end{equation}
all of which violate the sharp null hypothesis. In particular, $\bm P_1$ has different marginal probabilities; $\bm P_2$ and $\bm P_3$ have identical marginal probabilities but nonzero off diagonal elements.

Inspired by this example, to construct alternative hypotheses we introduce two measures quantifying violations of the sharp null hypothesis. We use the Hellinger distance $\tau_{HD}$ \citep{Hellinger:1909} to quantify the difference between the marginal probabilities:
\begin{equation*}
\tau_{HD} \left(\bm p_1, \bm p_0\right) =
\left\{\frac{1}{2}\sum\limits_{j=0}^{J-1}\left(p_{j+}^{1/2}-p_{+j}^{1/2}\right)^2\right\}^{1/2}.
\end{equation*}
Other choices include the Kullback--Leibler divergence and total variance distance. Under the sharp null hypothesis $\tau_{HD} = 0,$ and therefore nonzero $\tau_{HD}$ implies violations of the sharp null hypothesis. However, $\tau_{HD}$ relies solely on the marginal probabilities, ignoring the joint distribution of the potential outcomes. For example, the probability matrices $\bm P_2$ and $\bm P_3$ in \eqref{eq:examples} violate the sharp null hypothesis although $\tau_{HD} = 0.$ To address this issue we use Cohen's Kappa \citep{Cohen:1960}:
\begin{equation}
\kappa \left( \bm P \right) = \left\{\mathrm{tr}(\bm{P})-\bm{p}_1^\mathrm{T} \bm{p}_0\right\}
\left./\right.
\left(1-\bm{p}_1^\mathrm{T} \bm{p}_0\right), \label{eq:kappa_def}
\end{equation}
where $\mathrm{tr}(\cdot)$ is the trace function. Cohen's $\kappa$ relies on the probability matrix $\bm P,$ and under the sharp null hypothesis $\kappa=1$ because $\bm P$ is diagonal.

We now construct sequences of alternative hypotheses by varying the Hellinger distance $\tau_{HD}$ and Cohen's $\kappa$. To be more specific, we follow a two-step procedure:
\begin{enumerate}
\item construct a sequence of marginal probabilities, in an increasing order of $\tau_{HD};$
\item for fixed marginal probabilities, construct a sequence of probability matrices in an increasing order of $\kappa$, which involves the following sub-steps:
\begin{enumerate}
\item minimize and maximize $\kappa$ subject to the following constraints:
\begin{equation*}
\sum_{k^\prime =0}^{J-1}p_{k^\prime l} = p_{+l}, \quad
\sum_{l^\prime=0}^{J-1}p_{kl^\prime} = p_{k+}, \quad
p_{kl} \ge 0
\quad
(k, l = 0, 1, \ldots, J-1);
\end{equation*}
\item use a convex combination of the minimizer and maximizer to construct probability matrices with intermediate values of $\kappa$.
\end{enumerate}
\end{enumerate}
Step 1 helps access the impact of $\tau_{HD}$ on the powers of randomization tests, and Step 2 further helps access the impact of $\kappa.$ For fixed marginal probabilities, Sub-step (a) studies the two extreme cases of ``most'' and ``least'' violations of the sharp null hypothesis, and Sub-step (b) addresses the ``in between'' cases. Therefore, this procedure provides a relatively complete picture of violations of the sharp null hypothesis.

For given marginal probabilities $\bm p_1$ and $\bm p_0,$ the minimization problem in the above procedure is somewhat intuitive. Consider the probability matrix $\bm P_I$ with independent potential outcomes, i.e.,
$
p_{kl}=p_{k+} p_{+l}
$
for all
$
k,l.
$
If we are not interested in distributions with negatively associated potential outcomes, $\bm P_I$ minimizes $\kappa$ as zero. The maximization problem is, however, non-trivial, and almost intractable unless some restrictions are imposed on ${\bm P}$.  For the purpose of simplification, we make the following assumption:

\begin{assumption} \label{assume:sd}
(Stochastic Dominance) For all $j=1, \ldots, J-1$, $\sum_{k=j}^{J-1}p_{k+}\geq \sum_{l=j}^{J-1} p_{+l}$.
\end{assumption}

As an illustration, consider the three probability matrices in \eqref{eq:examples}. Among these matrices, ${\bm P}_1$ does not satisfy stochastic dominance. Besides the advantage of reducing the number of possible alternative hypotheses, in applied research the stochastic dominance pattern occurs frequently \citep[e.g.,][]{Bradley:1962, Bajorski:1999}, and is termed ``positive distributional causal effect'' in \cite{Ju:2010}. Because of the aforementioned technical convenience and the practical importance, we first focus on marginal probabilities $\bm p_1$ and $\bm p_0$ that satisfy Stochastic Dominance, and then discuss general marginal probabilities.

We further simplify the maximization problem by restricting the maximizer to be lower triangular by utilizing a well-known result that for any marginal probabilities satisfying  Stochastic Dominance, there exists a probability matrix that is lower triangular. This existence result is a special case of Strassen's theorem \citep{Strassen:1965, Lindvall:1992}, which was utilized by \cite{Rosenbaum:2001}. We are now in a position to state and prove the following theorem that provides the maximum value of $\kappa$, and the maximizer itself:

\begin{theorem} \label{thm:m}
For any $J\ge 2$, given marginal probabilities $\bm{p}_1$ and $\bm{p}_0$ satisfying Stochastic Dominance, there exists a lower triangular probability matrix $\bm P_+$ achieving the upper bound of $\kappa:$
\begin{equation}\label{eq:kmax}
\kappa \left( \bm P \right) \le \kappa \left( \bm P_+ \right) = \left\{\sum_{k=0}^{J-1}\min \left( p_{k+}, p_{+k} \right) - \bm{p}_1^\mathrm{T}\bm{p}_0\right\}/\left(1-\bm{p}_1^\mathrm{T}
\bm{p}_0\right).
\end{equation}
\end{theorem}

%\noindent {\it Proof of Theorem \ref{thm:m}: }
\begin{proof}[Proof of Theorem \ref{thm:m}]
The proof consists of two parts. We first show that \eqref{eq:kmax} is an upper bound of $\kappa,$ and then construct a probability matrix attaining it.

For all $k = 0, \ldots, J-1$, the diagonal element $p_{kk}$ of matrix ${\bm P}$ cannot be greater than either $p_{k+}$ or $p_{+k}$, i.e., $p_{kk} \le \min \left( p_{k+}, p_{+k} \right),$ which implies
\begin{equation}
\mathrm{tr} \left( \bm P \right) \le \sum_{k=0}^{J-1}\min \left( p_{k+}, p_{+k} \right). \label{eq:thm1_tr}
\end{equation}
Substituting \eqref{eq:thm1_tr} in \eqref{eq:kappa_def} yields the upper bound of $\kappa$ in \eqref{eq:kmax}.

We then sequentially construct a $J \times J$ lower triangular matrix with fixed row rums $\bm p_1$ and column sums $\bm p_0.$ We start with the last column and proceed backwards. At any point in the construction, we denote the element in the matrix already filled by $\tilde{p}_{kl}$ and those have not by $p_{kl}.$ First, for the last column with index $J-1$, only the last entry needs to be filled, and we set it equal to the corresponding column sum, i.e., $\tilde{p}_{J-1, J-1} = p_{+, J-1}$. Next, for all $r=1, \ldots, J-1$, given all elements in the last $r$ columns are already filled, we consider the problem of filling in the elements of column with index $l = J-r-1$, as shown in Table \ref{t:proof1}. At this point, the already filled elements in the matrix are $\tilde{p}_{kl}$, where $k = 0, \ldots, J-1$ and $l = J-r, \ldots, J-1$.

\begin{table}[h]
\centering\small
\caption{Filling in the column with index $l = J-r-1$ when the last $r$ columns are already filled} \label{t:proof1}
\begin{equation*}
\begin{array}{c|cc||c||ccc|c}
 \mbox{Row index} & \multicolumn{6}{|c|}{\mbox{Column index} \ (l)} & \mbox{Row Sum} \\
(k) & 0  & \cdots & J-r-1 & J-r & \cdots & J-1 & p_{k+} \\ \hline
0 & p_{00} & \cdots & 0   & 0 & \cdots &  0 & p_{0+} \\
\vdots & \vdots & \vdots & \vdots & \vdots & \vdots &  \vdots & \vdots \\
J-r-1 & p_{J-r-1, 0} & \cdots & \min(p_{J-r-1, +}, p_{+, J-r-1}) & 0 & \cdots & 0  & p_{J-r-1, +} \\
J-r &  p_{J-r, 0} & \cdots & p_{J-r,J-r-1} = ?                                        & \tilde{p}_{J-r,J-r} & \cdots & 0 & p_{J-r, +} \\
\vdots & \vdots & \vdots & \vdots & \vdots & \vdots &  \vdots & \vdots \\
J-1 &  p_{J-1, 0} & \cdots & p_{J-1,J-r-1}=?    & \tilde{p}_{J-1,J-r} & \cdots & \tilde{p}_{J-1,J-1}=p_{+,J-1} & p_{J-1, +} \\ \hline
\mbox{Column sum} & p_{+0} &  \cdots & p_{+, J-r-1} & p_{+, J-r} & \cdots & p_{+, J-1} & 1 \\
\end{array}
\end{equation*}
\end{table}

To fill the column with index $l = J-r-1$, note that all entries for $k<J-r-1$ will be equal to zero. We set the diagonal element with row index $k = l = J-r-1$ to be the minimum of the corresponding row and column sums, i.e., $\tilde{p}_{J-r-1, J-r-1} = \min \left( p_{J-r-1, +}, p_{+, J-r-1} \right)$. Now the difference $p_{+, J-r-1} -  \min (p_{J-r-1, +}, p_{+, J-r-1})$ needs to be distributed over the remaining entries of the column. Note that this difference is zero if $\min \left( p_{J-r-1, +}, p_{+, J-r-1} \right) = p_{+, J-r-1}$. Therefore, for all $k=J-r, \ldots, J-1$, we make the entry $p_{k, J-r-1}$ proportional to the ``remaining balance'' $p_{+, J-r-1} -  \min (p_{J-r-1, +}, p_{+, J-r-1})$, where we choose the proportionality constant as
\begin{equation}
\frac{\sum_{l=0}^{J-r-1} p_{kl}}{\sum_{k^\prime \ge J-r} \sum_{l=0}^{J-r-1} p_{k^\prime l}}, \label{eq:prop}
\end{equation}
that is, the ratio of the sum of empty entries in the row with label $k$ and the sum of empty entries in all rows below the one labeled $J-r-1$. Both the numerator and denominator of \eqref{eq:prop} can be expressed in terms of the given marginals and the already filled-in entries in the last $r$ columns:
\begin{equation}\label{eq:propfill}
\sum_{l=0}^{J-r-1} p_{kl} = p_{k+} - \sum_{l \ge J-r} \tilde{p}_{kl},
\quad
\sum_{k^\prime \ge J-r} \sum_{l=0}^{J-r-1} p_{k^\prime l} = \sum_{k^\prime \ge J-r} \left( p_{k^\prime +} - \sum_{l \ge J-r} \tilde{p}_{k^\prime l} \right)
\end{equation}
and hence can be computed uniquely. The construction method, \eqref{eq:prop} or \eqref{eq:propfill}, eventually leads to the following iterative imputation equation for all $k=J-r, \ldots, J-1$:
\begin{equation}
\tilde{p}_{k, J-r-1} = \frac{p_{k+} - \sum_{l \ge J-r} \tilde{p}_{kl}}{\sum_{k^\prime \ge J-r} \left( p_{k^\prime +} - \sum_{l \ge J-r} \tilde{p}_{k^\prime l} \right)} \left\{ p_{+, J-r-1} -  \min (p_{J-r-1, +}, p_{+, J-r-1}) \right\}. \label{eq:imp}
\end{equation}

We need to show the constructed matrix
\begin{equation*}
\bm P_+  = \left(\begin{array}{cccc}
\tilde{p}_{00} & 0 & \ldots &  0 \\
\tilde{p}_{10} & \tilde{p}_{11} & \ldots &  0 \\
\vdots & \vdots & \ddots & \vdots \\
\tilde{p}_{J-1, 0} & \tilde{p}_{J-1, 1} & \ldots & \tilde{p}_{J-1, J-1} \\
\end{array} \right)
\end{equation*}
indeed satisfies  (i) $\tilde{p}_{kl} \ge 0$ for all $k,l = 0, \ldots, J-1$, (ii) the equality in condition (\ref{eq:thm1_tr}), for which a sufficient condition is $\tilde{p}_{kk} = \min \left( p_{k+}, p_{+k} \right)$ for all $k = 0, \ldots, J-1$, (iii) the vector of column sums is ${\bm p}_0$, i.e., $\sum_{k=0}^{J-1} \tilde{p}_{kl} = p_{+l}$ for all $l = 0, \ldots, J-1$, and (iv) the vector of row sums is ${\bm p}_1$, i.e., $\sum_{l=0}^{J-1} \tilde{p}_{kl} = p_{k+}$ for all $k = 0, \ldots, J-1$.

Among (i)--(iv) described above, (i)--(iii) follow directly by the construction of ${\bm P}_{+}$ described above. We need only to prove that $\sum_{l=0}^{J-1} \tilde{p}_{kl} = p_{k+}$ for all $k = 0, \ldots, J-1$. By Stochastic Dominance, we have $p_{0+} \le p_{+0}$, implying that $\tilde{p}_{00} = p_{0+}$. Therefore the sum of the first row of $\bm P_+$ is $p_{0+}$. Now for all $k=1, \ldots, J-1$, by substituting $r=J-1$ in (\ref{eq:imp}), or by filling up the first column given the last $J-1,$ we have
\begin{eqnarray*}
\tilde{p}_{k0} &=& \frac{p_{k+}- \sum_{l=1}^{J-1} \tilde{p}_{kl}}{\sum_{k=1}^{J-1} \left( p_{k+} - \sum_{l=1}^{J-1} \tilde{p}_{kl} \right)} \left(p_{+0} - p_{0+} \right)
\;\: = \;\: \frac{p_{k+} - \sum_{l=1}^{J-1} \tilde{p}_{kl}}{(1-p_{0+}) - (1-\tilde{p}_{00}) } \left(p_{+0} - p_{0+} \right) \\
&=& \frac{p_{k+} - \sum_{l=1}^{J-1} \tilde{p}_{kl}}{p_{+0} -p_{0+} } \left(p_{+0} - p_{0+} \right)
\;\: = \;\: p_{k+} - \sum_{l=1}^{J-1} \tilde{p}_{kl},
\end{eqnarray*}
which implies that $\sum_{l=0}^{J-1} \tilde{p}_{kl} = p_{k+}.$ The proof is complete.
%$\;\;\Box$
\end{proof}

In the above proof of Theorem \ref{thm:m}, we suggest a way to constructing the maximizer $\bm P_+.$ Next we discuss the uniqueness of $\bm P_+$. By restricting $\bm P_+$ to be lower triangular and its $(j+1)$th diagonal element $p_{jj}$ to be $\min \left( p_{j+}, p_{+j} \right),$ what remain to be determined are the $(J-1)J/2$ off diagonal elements. Note that there are $(2J-3)$ constraints associated with them. The equality
$
(J-1)J / 2 = 2J-3
$
holds if and only if $J=2$ or $3$.

The case with $J=2$ corresponds to binary outcomes, which occur frequently in both methodology and applied research. For a recent discussion of finite population inference for binary data, see \cite{Ding:2015}. The following corollary provides the maximizer under $J=2$. Although it is a special case of Theorem \ref{thm:m}, we provide a direct proof to rigorously show the uniqueness of the maximizer.
\begin{corollary}\label{coro:m2}
For $J = 2$, given marginal probabilities $\bm{p}_1$ and $\bm{p}_0$ that satisfy Stochastic Dominance, the following matrix is the unique maximizer of $\kappa$:
\begin{equation}\label{eq:m2}
\bm P_+ =
\left(\begin{array}{cc}
p_{0+} &  0 \\
p_{1+} - p_{+1} & p_{+1} \\
\end{array} \right).
\end{equation}
\end{corollary}

%\noindent {\it Proof of Corollary \ref{coro:m2}: }
\begin{proof}[Proof of Corollary \ref{coro:m2}]
Because $\bm{p}_1$ and $\bm{p}_0$ satisfy Stochastic Dominance, we have $p_{0+} \le p_{+0}$ and $p_{1+} \ge p_{+1},$ implying that the diagonal elements of the maximizer are $p_{00} = p_{0+}$ and $p_{11} = p_{+1}$. Because the row sums of the maximizer are $\bm p_1$, we uniquely determine the entries of the maximizer, as shown in \eqref{eq:m2}. The maximizer has nonnegative entries because $p_{1+} \ge p_{+1}$, and its column sums are $\bm p_0$ because $p_{0+} + p_{1+} - p_{+1} = p_{+0}$. The proof is complete.
%$\;\;\Box$
\end{proof}

The case with $J=3$ corresponds to three-level outcomes, which are also important in practice. For example, in a clinical trial we can describe the status of a patient as ``deterioration,'' ``no change'' or ``improvement'' \citep{Bajorski:1999}. The following corollary gives the maximizer for $J=3$. Again, we provide a direct proof.
\begin{corollary}\label{coro:m3}
For $J = 3$, given marginal probabilities $\bm{p}_1$ and $\bm{p}_0$ that satisfy Stochastic Dominance, the following matrix is the unique maximizer of $\kappa$:
\begin{equation}\label{eq:m3}
\bm{P}_+ =
\left(\begin{array}{ccc}
p_{0+} &  0  & 0 \\
p_{1+} - \min\left(p_{+1}, p_{1+}\right) & \min\left(p_{+1}, p_{1+}\right) & 0\\
p_{2+}-p_{+2}-\left\{p_{+1}-\min\left(p_{+1}, p_{1+}\right)\right\} & p_{+1}-\min\left(p_{+1}, p_{1+}\right) & p_{+2} \\
\end{array} \right).
\end{equation}
\end{corollary}

%\noindent {\it Proof of Corollary \ref{coro:m3}: }
\begin{proof}[Proof of Corollary \ref{coro:m3}]
Because $\bm{p}_1$ and $\bm{p}_0$ satisfy Stochastic Dominance, we have
$
p_{0+} \le p_{+0}
$
and
$
p_{2+} \ge p_{+2},
$
which implies that the diagonal elements of the maximizer are $p_{00} = p_{0+}$, $p_{11} = \min\left(p_{+1}, p_{1+}\right)$, and $p_{22} = p_{+2}$. First, because the first row sum and third column sum are respectively $p_{0+}$ and $p_{+2}$, the maximizer is in the following form:
\begin{equation*}
\bm{P}_+ =
\left(\begin{array}{ccc}
p_{0+} &  0  & 0 \\
? & \min\left(p_{+1}, p_{1+}\right) & 0\\
? & ? & p_{+2} \\
\end{array} \right),
\end{equation*}
where ``?'' denotes an entry yet to be determined. Second, because the second row sum and column sum are respectively $p_{1+}$ and $p_{+1}$, the maximizer is in the following form:
\begin{equation*}
\bm{P}_+ =
\left(\begin{array}{ccc}
p_{0+} &  0  & 0 \\
p_{1+} - \min\left(p_{+1}, p_{1+}\right) & \min\left(p_{+1}, p_{1+}\right) & 0\\
? & p_{\cdot,
1}-\min\left(p_{+1}, p_{1+}\right) & p_{+2} \\
\end{array} \right).
\end{equation*}
Third, because the third row sum is $p_{2+}$, we uniquely determine the maximizer, as in \eqref{eq:m3}. Fourth, $\bm P_+$ has nonnegative entries, because
$$
p_{2+}-p_{+2}-\left\{p_{+ 1}-\min\left(p_{+ 1}, p_{1+}\right)\right\}
= \min\left(p_{2+}-p_{+ 2}, p_{+ 0}-p_{0+}\right)
\ge 0.
$$
Finally, $\bm P_+$ has row sums $\bm p_1$, and column sums $\bm p_0$, because
$$
p_{0+} + p_{1+} - \min\left(p_{+ 1}, p_{1+}\right) + p_{2+}-p_{+ 2}-\left\{p_{+ 1}-\min\left(p_{+ 1}, p_{1+}\right)\right\} = p_{+ 0}.
$$
The proof is complete.
%$\;\;\Box$
\end{proof}

We end this section by constructing probability matrices with intermediate values of $\kappa$. Given the minimizer $\bm{P}_I$ and maximizer $\bm{P}_+,$ let
$
\bm{P}_\lambda = \lambda \bm P_I + (1-\lambda)\bm P_+.
$
We view $\lambda \in [0, 1]$ as a sensitivity parameter, because we cannot estimate it from the observed data. The resulting probability matrices have the same marginal probabilities as $\bm P_I$ and $\bm P_+,$ and subsequently the same Hellinger distances. However, they have different $\kappa$ depending on $\lambda$ because
$
\kappa \left( \bm P_\lambda \right) = (1-\lambda)\kappa \left( \bm P_+ \right).
$
For the infinite population framework, our constructed sequence of alternative hypotheses are thus $\{\bm P_\lambda\}_{\lambda\in[0, 1]}$.
For the finite population framework, because any entry of a well-defined probability matrix $\bm P$ is multiples of $1/N,$ we propose a calibration step by letting
\begin{equation*}
\tilde{p}_{kl}(\lambda) =
\begin{cases}
    \frac{\lfloor Np_{kl}(\lambda)\rfloor}{N} & \mbox{if } k \ne l, \\
    p_{+ l} - \sum_{k^\prime \ne l} \frac{\lfloor Np_{k^\prime l}(\lambda)\rfloor}{N} & \mbox{if } k = l,
\end{cases}
\end{equation*}
where $\lfloor\cdot\rfloor$ is the floor function. By definition, the column sums of $\widetilde{\bm P}_\lambda$ are $\bm p_0$. Let $\Lambda\left(\bm p_1, \bm p_0\right)$ denote the set containing all $\lambda$'s such that the row sums of $\widetilde{\bm P}_\lambda$ are $\bm p_1$, and our constructed sequence of alternative hypotheses are therefore $\{\widetilde{\bm P}_\lambda\}_{\lambda\in \Lambda\left(\bm p_1, \bm p_0\right)}$. In practice, we can use a grid search to obtain an approximation of $\Lambda\left(\bm p_1, \bm p_0\right)$.

\section{A Simulation Study}\label{sec:simu}

We demonstrate how the above construction facilitates assessment of powers of randomization tests. We use the squared Mann--Whitney $U$-statistic (\citealp{Agresti:2002}):
\begin{equation*}
U^2 =
%\frac{1}{4N_1^2(N-N_1)^2}
\left[\sum\limits_{k=0}^{J-1}\sum\limits_{l=0}^{J-1}n_{1k}n_{0l}\left\{I(k>l)
- I(k<l)\right\}\right]^2.
\end{equation*}
Another commonly-used test statistic for categorical data is the $\chi^2$-statistic. However, we do not use it in this paper, because it does not utilize the order information and therefore is less powerful.

Although closed-form expressions of the powers of randomization test using the $U^2$ statistic are difficult to obtain, numerical calculations by Monte Carlo are straightforward, once we determine the alternative hypothesis $\bm P.$ We will follow three steps:
\begin{enumerate}
\item under $\bm P$ generate $2\times 10^5$ independent treatment assignments and obtain the corresponding observed data sets;
\item for each observed data set calculate the $p$-value of the randomization test using the observed value of $U^2$ and its simulated null distribution;
\item approximate the power of the $U^2$ statistic as the proportion of the $p$-values that are smaller than the significance level $\alpha = 0.05$.
\end{enumerate}

In this simulation study, we construct alternative hypotheses using the following four sets of marginal probabilities, two with $J=2$ and two with $J=3$:
\begin{enumerate}
    \item $\bm p_1 = (3/10, 7/10)^\mathrm{T}$, $\bm p_0 = (3/5, 2/5)^\mathrm{T}$, $\tau_{HD} = 0.216$;
    \item $\bm p_1 = (1/2, 1/2)^\mathrm{T}$, $\bm p_0 = (4/5, 1/5)^\mathrm{T}$, $\tau_{HD} = 0.227$;
    \item $\bm p_1 = (1/4, 1/4, 1/2)^\mathrm{T}$, $\bm p_0 = (2/5, 2/5, 1/5)^\mathrm{T}$, $\tau_{HD} = 0.227$;
    \item $\bm p_1 = (9/40, 9/40, 11/20)^\mathrm{T}$, $\bm p_0=(2/5, 2/5, 1/5)^\mathrm{T}$, $\tau_{HD} = 0.261$.
\end{enumerate}
For each case, we let the sample sizes $N=120, 160, 240$, and the sensitivity parameters $\lambda=0, 1/4, 1/2, 3/4, 1$. We then construct the probability matrices, which share the same marginals. For each probability matrix $\widetilde{\bm P}_\lambda,$ we use the aforementioned Monte Carlo procedure to calculate the powers. On one hand, different cases allow us to study the impact of $\tau_{HD}$ on the powers. On the other hand, within each case we can study the impact of $\kappa$ on the powers. The simulation results are summarized in Figure \ref{fg:power}, from which we draw the following conclusions. For all fixed sample sizes, the power functions of Case 2 dominate those of Case 1, and the power functions of Case 4 dominate those of Case 3. Therefore, for fixed $J$ the power increases as the Hellinger distance increases. Furthermore, for fixed marginals and sample size, the power increases as $\kappa$ decreases, or equivalently as $\lambda$ increases. However, this dependence becomes weaker as the sample size increases, because the power converges to one.

We can use the demonstrated methodology to compare the power functions of different test statistics, and also to determine sample sizes that guarantee a pre-specified power. For instance, in Case 3, we cannot guarantee a power of 0.95 with a sample of size 120, but we can with a sample of size 160.

\begin{figure}[htbp]
\centering
 \includegraphics[height=.9\linewidth, width=.9\linewidth]{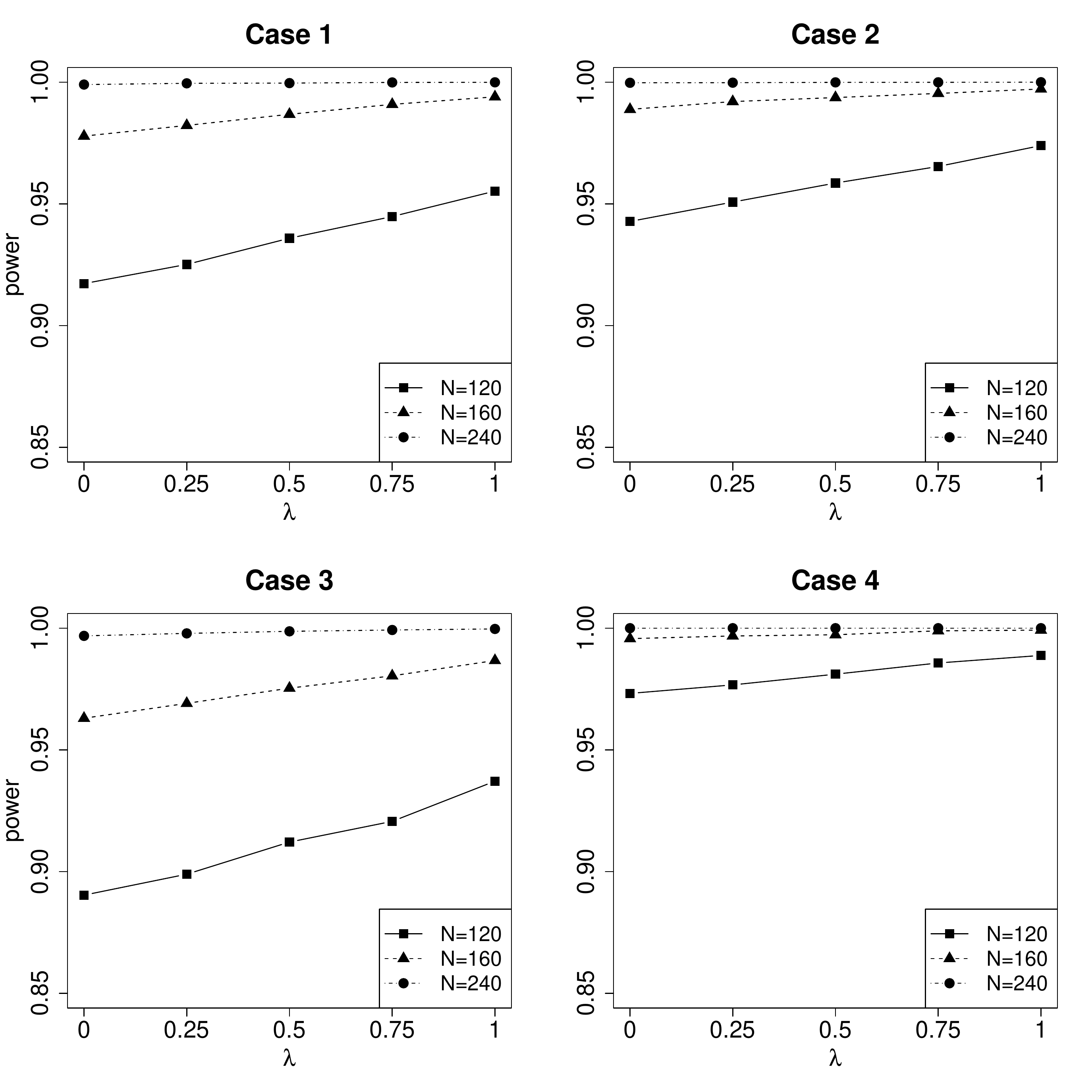}
\caption{Statistical Powers of Randomization Tests Using $U^2$.}
\label{fg:power}
\end{figure}

In summary, for a finite population, the power of the randomization test using $U^2$ depends on the marginal difference of the potential outcomes as well as the association between them. In particular, the power increases as the marginal difference increases, and given the marginals fixed, the power increases as the association decreases. Furthermore, the power converges to one as the sample size increases, for any case with nonzero marginal difference. The above conclusions appear to confirm our intuition, because it should be easier to reject the sharp null hypothesis given larger differences between the marginals, and given the marginals fixed, it should be easier to reject the sharp null hypothesis given smaller associations between the potential outcomes. Our findings conform to \cite{Plackett:1977}'s and \cite{Chernoff:2004}'s results about the classical $2\times 2$ tables: the marginals of the contingency tables contain limited amount of information about the association with finite samples, which becomes negligible asymptotically.

\section{Discussion}\label{sec:dis}

In this paper, we construct sequences of finite populations of ordinal outcomes in increasing departures from the sharp null hypothesis of no treatment effect. Our construction procedure is useful for evaluating of the powers of randomization tests. In particular, our construction procedure takes into account the dependence structure of the potential outcomes, whereas existing literature often assume independent potential outcomes. Through a simulation study, we demonstrate that the association between potential outcomes indeed affects the powers of randomization tests. We argue that taking into account the association between potential outcomes is crucial for conducting randomization tests in practice, for example when determining sample sizes that guarantee a pre-specified power.

There are multiple future directions based on our work. First, although we adopt a numerical approach, it is possible to derive the asymptotic distribution of the $U^2$ statistic under the sharp null hypothesis. Second, we can derive the maximizer of $\kappa$ for general marginal probabilities that do not satisfy Stochastic Dominance. Third, we can incorporate covariate information to further improve the powers of randomization tests. Fourth, while the Fisherian randomization-based inference is a useful first step, Neymanian and Bayesian counterparts of causal inference for ordinal outcomes are still needed. For some recent developments, see \cite{Lu:2015}.

\section*{Acknowledgements}

The authors thank a reviewer and the editor for their thoughtful comments that helped clarify several points better and substantially improved the presentation of the paper.

\bibliographystyle{apalike}
\bibliography{ordinal2015}

\end{document}